\documentclass[11pt, reqno]{amsart}

\usepackage{amsmath, amsthm, amscd, amsfonts, amssymb, graphicx, color}
\usepackage[bookmarksnumbered, colorlinks, plainpages]{hyperref}
\usepackage{dsfont}
\input{mathrsfs.sty}
\usepackage{subfigure}
\usepackage{bbm,mathtools}
\usepackage{booktabs}
\usepackage{enumitem}
 \makeatletter
\let\reftagform@=\tagform@
\def\tagform@#1{\maketag@@@{(\ignorespaces\textcolor{blue}{#1}\unskip\@@italiccorr)}}
\renewcommand{\eqref}[1]{\textup{\reftagform@{\ref{#1}}}}
\makeatother
\usepackage{hyperref}
\hypersetup{colorlinks=true, linkcolor=red, anchorcolor=green,
citecolor=cyan, urlcolor=red, filecolor=magenta, pdftoolbar=true}

\newtheorem{theorem}{Theorem}[section]
\newtheorem{lemma}[theorem]{Lemma}
\newtheorem{proposition}[theorem]{Proposition}
\newtheorem{corollary}[theorem]{Corollary}
\theoremstyle{definition}
\newtheorem{definition}[theorem]{Definition}
\newtheorem{example}[theorem]{Example}

\theoremstyle{remark}
\newtheorem{remark}[theorem]{Remark}
\numberwithin{equation}{section}

\newcommand{\uinorm}[1]{{\left\vert\kern-0.25ex\left\vert\kern-0.25ex\left\vert #1
 \right\vert\kern-0.25ex\right\vert\kern-0.25ex\right\vert}}

\newcommand{\Tr}{\mathrm{Tr}}
\newcommand{\ran}{\mathrm{ran}}

\begin{document}

\setcounter{page}{1}

\title[Finite-Corner Reconstruction]{Finite-Corner Reconstruction  and Separation for Cones of Normal Maps}

\author[M. Kian and F. Kittaneh]{Mohsen Kian$^1$ and Fuad Kittaneh$^{2,3}$ }

\address{$^1$ Mohsen Kian: Department of Mathematics, University of Bojnord, P.O. Box
1339, Bojnord 94531, Iran}
\email{kian@ub.ac.ir }

\address{$^2$ Fuad Kittaneh: Department of Mathematics,  The University of Jordan, Amman, Jordan}
\address{$^3$ Fuad Kittaneh: Department of Mathematics,  Korea University, Seoul 02841, South Korea}
\email{fkitt@ju.edu.jo}



\begin{abstract}
We develop a finite-corner framework for convex cones of normal maps between \(\mathcal B(\mathcal H)\) and \(\mathcal B(\mathcal K)\). The basic structural assumption is stability under finite cut--pad operations. We prove that every point-ultraweakly closed cut--pad stable cone is completely determined by its finite-dimensional corner cones, and conversely that every coherent family of closed finite-dimensional corner cones admits a unique global realization of this type. For a cut--pad stable cone that is not necessarily point-ultraweakly closed, reconstruction from the norm closures of its corner cones yields exactly its point-ultraweak closure. Combining this reconstruction principle with finite-dimensional separation and the Choi representation, we show that non-membership in a global cone is always detected on a single finite-dimensional corner by a finite-dimensional witness. We illustrate the framework for completely positive and decomposable maps; in the decomposable case, the Hermitian parts of the finite-corner witnesses are PPT.
\end{abstract}
\maketitle

 \section{Introduction}
The study of linear maps between operator algebras is central to both
operator theory and quantum information. In finite dimensions, the
Choi--Jamio\l{}kowski correspondence associates a linear map
\[
\Phi:\mathbb{M}_n\to \mathbb{M}_k
\]
with an operator in $\mathbb{M}_n\otimes \mathbb{M}_k$
\cite{Jamiolkowski1972,Choi1975}. In particular, Choi's theorem
characterizes complete positivity of $\Phi$ by positivity of its Choi
matrix \cite{Choi1975}. In quantum information, the  idea  is
closely related to the positive-map criterion for separability
\cite{Horodecki1996} and to entanglement-witness formulations
\cite{Terhal2000}.

In infinite dimensions, the same picture is more delicate. Holevo
formulated an infinite-dimensional Choi--Jamio\l{}kowski correspondence
in terms of positive semidefinite forms on a suitable dense subspace,
which in general need not be closable \cite{Holevo2011CJ}. Thus a
bounded global Choi operator is not available in general in the same
way as for maps between matrix algebras. This suggests replacing a global Choi object by the ordinary Choi matrices of finite-dimensional compressions. 

Mapping cones provide a related structural setting.  In the
matrix-algebra setting, mapping cones are closed cones of positive maps
that are invariant under pre- and post-composition by completely
positive maps. Johnston and St{\o}rmer showed that every mapping cone
determines a unique operator system with a certain homogeneity property,
and conversely that operator systems with this property give rise to
mapping cones \cite{Johnston2012}. In a different, infinite-dimensional
module setting, Magajna associated Choi-type operators with
weak$^*$-continuous bimodule maps, studied cones of such maps together
with corresponding dual cones, and showed that, in an appropriate
setting, these notions reduce to those considered earlier by
St{\o}rmer \cite{Magajna2021}.

Infinite-dimensional Choi representations have also been developed   in recent works. Han, Kye, and St{\o}rmer associate, for a
class of linear maps on a von Neumann factor, bounded and trace-class
operators that play the role of Choi matrices, and use these objects to
characterize complete positivity and study positivity
\cite{Han2024}. Vom Ende develops a modified Choi formalism on
separable Hilbert spaces and uses it to obtain unique decompositions of
generators of norm-continuous semigroups of completely positive maps
\cite{Ende2024}.

Positive maps have also been studied through finer notions of positivity and decomposability. Osaka classified families of positive
maps in low-dimensional matrix algebras according to their degree of
indecomposability and obtained results concerning atomicity
\cite{Osaka1991}. He subsequently constructed examples of positive maps
that cannot be expressed as the sum of a $2$-positive map and a
$2$-copositive map \cite{Osaka1993}. More recently, M{\l}ynik, Osaka,
and Marciniak characterized $k$-positivity in terms of Ky Fan norm
estimates associated with Kraus operators and applied this
characterization to parameterized families of positive maps, including
the determination of decomposability regions \cite{Mlynik2025}. In the
$C^*$-algebraic setting, Bhat and Osaka established factorization
results for positive maps and showed, in particular, that if
$\tau\otimes\operatorname{id}_k$ is decomposable for some $k\geq 2$,
then $\tau$ is completely positive \cite{BhatOsaka2020}.

Our starting observation     is that, although a bunded global Choi operator need not exist in infinite dimensions, every finite-dimensional compression of a normal map has an ordinary Choi matrix. It is therefore natural to ask whether the collection of all such finite corners can serve as a substitute for a single global Choi representation. We consider this question for cones that are stable under finite cut--pad operations. The resulting problem is to determine when compatible finite-dimensional cone data determine the global cone.

We show that, under point-ultraweak closedness, this finite-corner information is sufficient: coherent systems of closed finite-dimensional corner cones correspond precisely to point-ultraweakly closed cut--pad stable cones of normal maps. For cones that are not point-ultraweakly closed, the corner data determine precisely the point-ultraweak closure. Combined with ordinary finite-dimensional separation, the reconstruction principle further implies that failure of global membership is already visible on a single finite-dimensional corner through a Choi-type witness. Thus finite corners serve two purposes: they reconstruct the global cone and provide finite-dimensional witnesses for non-membership.

 \section{Preliminaries}

In this section, we fix notation and review essential facts regarding operator topologies, normal maps, and the Choi--Jamio\l{}kowski isomorphism.

\subsection{Operator Spaces and Topologies}
Let $\mathcal H$ and $\mathcal K$ be complex Hilbert spaces. We denote by $\mathcal B(\mathcal H)$ the Banach space of bounded linear operators on $\mathcal H$ equipped with the operator norm $\|\cdot\|$, and by $\mathcal T(\mathcal H)$ the Banach space of trace-class operators equipped with the trace norm $\|T\|_1 := \Tr|T|$.

The space $\mathcal B(\mathcal H)$ is the isometric dual of $\mathcal T(\mathcal H)$ under the canonical bilinear pairing
\[
\langle \omega, X \rangle := \Tr(\omega X), \quad \omega \in \mathcal T(\mathcal H), \ X \in \mathcal B(\mathcal H).
\]
The \emph{ultraweak topology} (or weak$^*$ topology) on $\mathcal B(\mathcal H)$ is the topology induced by this pairing. A net $(X_\alpha)$ converges ultraweakly to $X$ if $\Tr(\omega X_\alpha) \to \Tr(\omega X)$ for all $\omega \in \mathcal T(\mathcal H)$.

We also refer to the \emph{strong operator topology} (SOT), where $X_\alpha \to X$ strongly if $\|(X_\alpha - X)\xi\| \to 0$ for all vectors $\xi \in \mathcal H$. Note that for a bounded net of operators, SOT convergence implies ultraweak convergence.

\subsection{Normal Linear Maps}
A linear map $\Phi: \mathcal B(\mathcal H) \to \mathcal B(\mathcal K)$ is called \emph{normal} if it is ultraweakly continuous. We denote the space of all normal maps by $\mathcal N(\mathcal B(\mathcal H), \mathcal B(\mathcal K))$.
We equip this space with the \emph{point-ultraweak topology}: a net of maps $(\Phi_\alpha)$ converges to $\Phi$ if
\[
\Phi_\alpha(X) \longrightarrow \Phi(X) \quad \text{ultraweakly in } \mathcal B(\mathcal K) \quad \text{for all } X \in \mathcal B(\mathcal H).
\]
For a subset $\mathsf S \subseteq
\mathcal N\bigl(\mathcal B(\mathcal H),\mathcal B(\mathcal K)\bigr)$, we write $\overline{\mathsf S}^{\,\text{puw}}$ for its closure in $\mathcal N\bigl(\mathcal B(\mathcal H),\mathcal B(\mathcal K)\bigr)$. equipped with the point-ultraweak topology.
\subsection{Finite Corners and the Choi Matrix}
Let $\mathcal P_{\mathrm{fin}}(\mathcal H)$ denote the set of finite-rank orthogonal projections on $\mathcal H$. We define the directed set of corners $\mathcal D$ as the collection of pairs $(Q,P) \in \mathcal P_{\mathrm{fin}}(\mathcal K) \times \mathcal P_{\mathrm{fin}}(\mathcal H)$, ordered by range inclusion:
\[
(Q,P) \preceq (Q',P') \iff \ran Q \subseteq \ran Q' \text{ and } \ran P \subseteq \ran P'.
\]
For a fixed corner $(Q,P) \in \mathcal D$, let $\mathcal B_P := \mathcal B(\ran P) \cong P\mathcal B(\mathcal H)P$. We denote the space of linear maps between these finite-dimensional algebras by
\[
\mathcal L_{Q,P} := \mathcal L(\mathcal B_P, \mathcal B(\ran Q)).
\]
To handle duality explicitly, we fix a system of matrix units $(E_{ij})_{i,j=1}^m$ for $\mathcal B_P$ (where $m = \dim\ran P$). The \emph{Choi matrix} of a map $\psi \in \mathcal L_{Q,P}$ is defined as
\[
C_\psi := \sum_{i,j=1}^m E_{ij} \otimes \psi(E_{ij}) \in \mathcal B(\ran P \otimes \ran Q).
\]
The map $\psi \mapsto C_\psi$ is a linear isomorphism (the Choi--Jamio\l{}kowski isomorphism). It is well-known that $\psi$ is completely positive (CP) if and only if $C_\psi \ge 0$.

 \subsection{Convex Geometry and Duality}
While $\mathcal B(\mathcal H)$ and $\mathcal T(\mathcal H)$ are complex Banach spaces, the cones of interest in this work (such as completely positive maps) are invariant only under non-negative real scalars, not arbitrary complex scalars. Consequently, the geometric analysis of these cones requires viewing the ambient spaces as real vector spaces (by restriction of scalars).

Accordingly, the appropriate dual pairing is the real part of the canonical pairing:
\[
\langle \omega, X \rangle_{\mathbb{R}} := \Re \Tr(\omega X), \quad \omega \in \mathcal T(\mathcal H), \ X \in \mathcal B(\mathcal H).
\]
Note that for self-adjoint operators $\omega, X$, this reduces to the standard trace pairing. Similarly, on the finite corners, we utilize the real inner product induced by the Hilbert--Schmidt norm:
\[
\langle \psi, \phi \rangle_{\mathcal L} := \Re \Tr(C_\psi^* C_\phi).
\]

\section{Finite-to-Infinite Reconstruction Theorem}

We now introduce finite compression and cut--pad operations and formulate the compatibility conditions for systems of corner cones. The main result of this section identifies coherent systems of closed finite-dimensional corner cones with point-ultraweakly closed cut--pad stable cones of normal maps.
\medskip

Let $P\in\mathcal B(\mathcal H)$ and $Q\in\mathcal B(\mathcal K)$ be finite-rank projections.
We identify the corner algebras
\[
\mathcal B(\ran P)\ \cong\ P\,\mathcal B(\mathcal H)\,P,
\qquad
\mathcal B(\ran Q)\ \cong\ Q\,\mathcal B(\mathcal K)\,Q,
\]
via the natural compression/extension by $P$ and $Q$.

\begin{definition}\label{def:corner-cutpad}
Let $\Phi\in\mathcal N(\mathcal B(\mathcal H),\mathcal B(\mathcal K))$ and let $P,Q$ be finite-rank projections.
\begin{enumerate}
\item[(a)] The \emph{compression} of $\Phi$ to the corner $(Q,P)$ is the linear map
\[
\Phi_{Q,P}:\mathcal B(\ran P)\to \mathcal B(\ran Q),
\qquad
\Phi_{Q,P}(X):=Q\,\Phi(X)\,Q,\quad X\in P\mathcal B(\mathcal H)P\cong \mathcal B(\ran P).
\]

\item[(b)] The \emph{cut--pad} map of $\Phi$ at $(Q,P)$ is the normal map
\[
\Phi^{(Q,P)}:\mathcal B(\mathcal H)\to\mathcal B(\mathcal K),
\qquad
\Phi^{(Q,P)}(T):=Q\,\Phi(PTP)\,Q,\quad T\in\mathcal B(\mathcal H).
\]
\end{enumerate}
\end{definition}

\begin{remark}
\label{rem:same-formula}
The formulas for $\Phi_{Q,P}$ and $\Phi^{(Q,P)}$ are identical, but the domains differ:
$\Phi_{Q,P}$ is defined only on the finite corner algebra $P\mathcal B(\mathcal H)P$,
whereas $\Phi^{(Q,P)}$ is a map on all of $\mathcal B(\mathcal H)$.
Moreover,
\[
\bigl(\Phi^{(Q,P)}\bigr)_{Q,P}=\Phi_{Q,P}.
\]
\end{remark}

\begin{definition}\label{def:cut-pad-stable}
A subset $\mathsf K\subseteq \mathcal N(\mathcal B(\mathcal H),\mathcal B(\mathcal K))$ is called a
\emph{cut--pad stable cone} if:
\begin{enumerate}
\item[(i)] $\mathsf K$ is a convex cone (closed under $+$ and multiplication by $\lambda\ge 0$);
\item[(ii)] for every $\Phi\in\mathsf K$ and every pair of finite-rank projections $P$ and $Q$,
\[
\Phi^{(Q,P)}\in \mathsf K.
\]
\end{enumerate}
We say $\mathsf K$ is \emph{point-ultraweakly closed} if it is closed in the point-ultraweak topology.
\end{definition}

\begin{remark} 
The cut--pad stability condition is automatic for cones that are stable
under pre- and post-composition by normal completely positive maps.
Indeed, for finite-rank projections $P$ and $Q$, the maps
\[
\beta_P(T)=PTP,
\qquad
\alpha_Q(Y)=QYQ
\]
are normal completely positive, and
\[
\Phi^{(Q,P)}=\alpha_Q\circ\Phi\circ\beta_P.
\]
Thus cut--pad stability is, in particular, weaker than requiring stability
under arbitrary normal completely positive pre- and post-compositions.
\end{remark}

\begin{definition}\label{def:corner-cone}
Let $\mathsf K$ be a cut--pad stable cone. For finite-rank projections $P,Q$ define its \emph{corner cone} by
\[
\mathsf K_{Q,P}:=\{\Phi_{Q,P}:\Phi\in\mathsf K\}.
\]
(Here $\mathsf K_{Q,P}$ is a cone of maps $\mathcal B(\ran P)\to\mathcal B(\ran Q)$; in finite dimension,
all maps are automatically normal.)
\end{definition}

\subsection*{Directed system of corners}

If $(Q,P)\preceq (Q',P')$, define the restriction map
\[
\rho_{(Q',P')\to(Q,P)}:
\mathcal L_{Q',P'}\longrightarrow\mathcal L_{Q,P}
\]
by
\[
\bigl(\rho_{(Q',P')\to(Q,P)}(\psi)\bigr)(X)
:=
Q\,\psi(X)\,Q,
\qquad
X\in\mathcal B(\ran P),
\]
where $\mathcal B(\ran P)$ is identified with the corresponding corner
of $\mathcal B(\ran P')$.

Similarly, if $(Q,P)\preceq (Q',P')$ and $\psi\in\mathcal L_{Q,P}$, define the \emph{padding (extension-by-zero)} map
\[
\iota_{(Q,P)\to(Q',P')}:\mathcal L_{Q,P}\to\mathcal L_{Q',P'},
\qquad
\bigl(\iota_{(Q,P)\to(Q',P')}(\psi)\bigr)(T):=Q'\,\psi(PTP)\,Q',
\quad T\in\mathcal B(\ran P').
\]

\subsection*{Coherent corner data and the induced infinite cone}
\begin{definition}\label{def:coherent-family}
A family $\mathcal C=\{\mathcal C_{Q,P}\subseteq \mathcal L_{Q,P}\}_{(Q,P)\in\mathcal D}$ is called \emph{coherent} if:
\begin{enumerate}
\item[(C1)] each $\mathcal C_{Q,P}$ is a closed convex cone in the finite-dimensional space $\mathcal L_{Q,P}$;
\item[(C2)] (\emph{restriction compatibility}) if $(Q,P)\preceq (Q',P')$, then
\[
\rho_{(Q',P')\to(Q,P)}(\mathcal C_{Q',P'})\subseteq \mathcal C_{Q,P};
\]
\item[(C3)] (\emph{padding compatibility}) if $(Q,P)\preceq (Q',P')$, then
\[
\iota_{(Q,P)\to(Q',P')}(\mathcal C_{Q,P})\subseteq \mathcal C_{Q',P'}.
\]
\end{enumerate}
\end{definition}

\begin{definition}\label{def:K-of-C}
Given a coherent family $\mathcal C$, define
\[
\mathsf K(\mathcal C)
:=
\Bigl\{\Phi\in \mathcal N(\mathcal B(\mathcal H),\mathcal B(\mathcal K)):\
\Phi_{Q,P}\in \mathcal C_{Q,P}\ \text{for all }(Q,P)\in\mathcal D\Bigr\}.
\]
\end{definition}

\medskip
We first record the trace-norm approximation needed in the reconstruction proof. 
\begin{lemma}\label{lem:trace-approx}
Let $\mathcal H$ be a Hilbert space and let $(P_\alpha)_\alpha$ be a net of finite-rank projections on $\mathcal H$
such that $P_\alpha\to I_{\mathcal H}$ strongly.
Then for every trace-class operator $S\in\mathcal T(\mathcal H)$,
\[
\|P_\alpha S P_\alpha - S\|_1 \longrightarrow 0.
\]
\end{lemma}

\begin{proof}
We begin with rank-one operators. Given $u,v\in\mathcal H$, define the rank-one operator
\[
u\otimes v^*:\mathcal H\to\mathcal H,\qquad (u\otimes v^*)(x):=\langle x,v\rangle\,u.
\]
It is well-known (and easy to check) that $u\otimes v^*$ is trace-class and
\begin{equation}\label{eq:rankone-tracenorm}
\|u\otimes v^*\|_1=\|u\|\,\|v\|.
\end{equation}
Moreover, for any $u,v,u',v'\in\mathcal H$ one has the estimate
\begin{equation}\label{eq:rankone-diff}
\|(u\otimes v^*)-(u'\otimes (v')^*)\|_1
\le \|(u-u')\otimes v^*\|_1+\|u'\otimes (v-v')^*\|_1
=\|u-u'\|\,\|v\|+\|u'\|\,\|v-v'\|,
\end{equation}
where we used \eqref{eq:rankone-tracenorm} in the last step.

Now suppose $S=u\otimes v^*$. Then $P_\alpha S P_\alpha=(P_\alpha u)\otimes (P_\alpha v)^*$, hence by
\eqref{eq:rankone-diff},
\[
\|P_\alpha S P_\alpha - S\|_1
=\|(P_\alpha u)\otimes (P_\alpha v)^* - u\otimes v^*\|_1
\le \|P_\alpha u-u\|\,\|P_\alpha v\|+\|u\|\,\|P_\alpha v-v\|.
\]
Since $P_\alpha\to I$ strongly, we have $P_\alpha u\to u$ and $P_\alpha v\to v$ in norm; also
$\|P_\alpha v\|\le \|v\|$. Therefore $\|P_\alpha S P_\alpha - S\|_1\to 0$ for rank-one $S$.

Next, if $S$ is finite-rank, write it as a finite sum of rank-one operators
$S=\sum_{j=1}^N u_j\otimes v_j^*$. Using linearity and the triangle inequality,
\[
\|P_\alpha S P_\alpha - S\|_1
\le \sum_{j=1}^N \|P_\alpha (u_j\otimes v_j^*) P_\alpha - (u_j\otimes v_j^*)\|_1 \longrightarrow 0,
\]
by the rank-one case.

Finally, let $S\in\mathcal T(\mathcal H)$ be arbitrary and fix $\varepsilon>0$.
Choose a finite-rank operator $S_0$ such that $\|S-S_0\|_1<\varepsilon$.
Since $P_\alpha$ is a contraction, $\|P_\alpha(S-S_0)P_\alpha\|_1\le \|S-S_0\|_1<\varepsilon$, hence
\[
\|P_\alpha S P_\alpha - S\|_1
\le \|P_\alpha(S-S_0)P_\alpha\|_1+\|P_\alpha S_0 P_\alpha - S_0\|_1+\|S_0-S\|_1
< 2\varepsilon+\|P_\alpha S_0 P_\alpha - S_0\|_1.
\]
Letting $\alpha\to\infty$ and using the finite-rank case gives
$\limsup_\alpha \|P_\alpha S P_\alpha - S\|_1 \le 2\varepsilon$.
Since $\varepsilon>0$ was arbitrary, we conclude $\|P_\alpha S P_\alpha - S\|_1\to 0$.
\end{proof}

We can now state the reconstruction theorem.
\begin{theorem}\label{thm:Reconstruction}
Let $\mathcal C=\{\mathcal C_{Q,P}\}_{(Q,P)\in\mathcal D}$ be a coherent family of corner cones.
Then $\mathsf K(\mathcal C)$ is a convex cut--pad stable cone of normal maps and is point-ultraweakly closed.

Moreover, $\mathsf K(\mathcal C)$ recovers the prescribed corner data: for every $(Q,P)\in\mathcal D$,
\[
\bigl(\mathsf K(\mathcal C)\bigr)_{Q,P}=\mathcal C_{Q,P}.
\]

Conversely, if $\mathsf K\subseteq \mathcal N(\mathcal B(\mathcal H),\mathcal B(\mathcal K))$ is a convex cut--pad stable 
  point-ultraweakly closed cone and $\mathcal C^{\mathsf K}:=\{\mathsf K_{Q,P}\}_{(Q,P)\in\mathcal D}$ denotes its family
of corner cones, then $\mathcal C^{\mathsf K}$ is coherent and
\[
\mathsf K=\mathsf K(\mathcal C^{\mathsf K}).
\]

In particular, the assignments
\[
\mathsf K\ \longmapsto\ \{\mathsf K_{Q,P}\}_{(Q,P)\in\mathcal D}
\qquad\text{and}\qquad
\mathcal C\ \longmapsto\ \mathsf K(\mathcal C)
\]
are inverse bijections between convex cut--pad stable  point-ultraweakly closed cones of normal maps and coherent families of corner
cones.
\end{theorem}

\begin{proof}
\medskip
\noindent{(A) We first show that $\mathsf K(\mathcal C)$ is a convex cut--pad stable point-ultraweakly closed cone.}
\smallskip

\emph{Convex cone.}
Let $\Phi_1,\Phi_2\in\mathsf K(\mathcal C)$ and $\lambda\ge 0$.
For every $(Q,P)\in\mathcal D$ we have $(\Phi_1)_{Q,P},(\Phi_2)_{Q,P}\in\mathcal C_{Q,P}$ and $\mathcal C_{Q,P}$ is a cone, hence
\[
(\Phi_1+\Phi_2)_{Q,P}=(\Phi_1)_{Q,P}+(\Phi_2)_{Q,P}\in\mathcal C_{Q,P},
\qquad
(\lambda\Phi_1)_{Q,P}=\lambda(\Phi_1)_{Q,P}\in\mathcal C_{Q,P}.
\]
Thus $\Phi_1+\Phi_2,\lambda\Phi_1\in\mathsf K(\mathcal C)$.

\emph{cut--pad stability.}
Fix $\Phi\in\mathsf K(\mathcal C)$ and fix a finite-rank pair $(Q_0,P_0)\in\mathcal D$.
We must show $\Phi^{(Q_0,P_0)}\in\mathsf K(\mathcal C)$, i.e.\ for every $(Q,P)\in\mathcal D$,
\begin{align}\label{eq:need-corner-final}
\bigl(\Phi^{(Q_0,P_0)}\bigr)_{Q,P}\in \mathcal C_{Q,P}.
\end{align}
So fix $(Q,P)\in\mathcal D$.
Choose finite-rank projections $Q'\ge Q,Q_0$ and $P'\ge P,P_0$ (e.g.\ onto $\ran(Q)+\ran(Q_0)$ and $\ran(P)+\ran(P_0)$).
Since $\Phi\in\mathsf K(\mathcal C)$, we have $\Phi_{Q_0,P_0}\in\mathcal C_{Q_0,P_0}$.
By padding coherence (C3),
\[
\iota_{(Q_0,P_0)\to(Q',P')}\bigl(\Phi_{Q_0,P_0}\bigr)\in \mathcal C_{Q',P'}.
\]
By restriction coherence (C2),
\begin{align}\label{eq:in-CQP-final}
\rho_{(Q',P')\to(Q,P)}\!\left(\iota_{(Q_0,P_0)\to(Q',P')}\bigl(\Phi_{Q_0,P_0}\bigr)\right)\in \mathcal C_{Q,P}.
\end{align}
We now compute that the map in \eqref{eq:in-CQP-final} is exactly the corner map in \eqref{eq:need-corner-final}.
It follows from definition of $\rho$ and $\iota$ that  for $X\in\mathcal B(\ran P)\subseteq \mathcal B(\ran P')$,
\begin{align*}
\Bigl[\rho_{(Q',P')\to(Q,P)}\!\left(\iota_{(Q_0,P_0)\to(Q',P')}\bigl(\Phi_{Q_0,P_0}\bigr)\right)\Bigr](X)
&=Q\cdot \Bigl[\iota_{(Q_0,P_0)\to(Q',P')}\bigl(\Phi_{Q_0,P_0}\bigr)\Bigr](X)\cdot Q\\
&=Q\cdot Q'\,\Phi_{Q_0,P_0}(P_0XP_0)\,Q'\cdot Q\\
&=Q\,\Phi_{Q_0,P_0}(P_0XP_0)\,Q\qquad(\text{since $QQ'=Q$})\\
&=Q\,Q_0\,\Phi(P_0XP_0)\,Q_0\,Q\\
&=\bigl(\Phi^{(Q_0,P_0)}\bigr)_{Q,P}(X),
\end{align*}
where the last equality follows directly from the definitions of the cut--pad map and the corner map.
Therefore \eqref{eq:in-CQP-final} implies \eqref{eq:need-corner-final}, proving $\Phi^{(Q_0,P_0)}\in\mathsf K(\mathcal C)$.

\emph{Ultraweak closedness.}
Let $(\Phi_i)$ be a net in $\mathsf K(\mathcal C)$ converging to a normal map $\Phi$ in the point-ultraweak topology,
i.e.\ $\omega(\Phi_i(T))\to\omega(\Phi(T))$ for all $T\in\mathcal B(\mathcal H)$ and all $\omega\in\mathcal T(\mathcal K)$.

Fix $(Q,P)\in\mathcal D$ and $X\in\mathcal B(\ran P)\subseteq\mathcal B(\mathcal H)$.
Let $\omega\in\mathcal T(\ran Q)$.  Via the canonical identification
\[
\mathcal B(\ran Q)_*\;\cong\;\mathcal T(\ran Q),\qquad
\omega(Z)=\Tr(\omega Z)\quad(Z\in\mathcal B(\ran Q)),
\]
we may regard $\omega$ as a normal linear functional on the corner algebra
$Q\mathcal B(\mathcal K)Q\cong \mathcal B(\ran Q)$.

To compare with point-ultraweak convergence on $\mathcal B(\mathcal K)$, extend $\omega$ to a
trace-class operator on $\mathcal K$ by zero outside $\ran Q$, namely
\[
\widetilde\omega:=Q\,\omega\,Q\ \in\ \mathcal T(\mathcal K).
\]
Equivalently, the associated normal functional on $\mathcal B(\mathcal K)$ is
\[
\widetilde\omega(Y):=\Tr(\widetilde\omega\,Y)\qquad (Y\in\mathcal B(\mathcal K)).
\]
Since $\widetilde\omega$ is supported on $\ran Q$, we have for every $Y\in\mathcal B(\mathcal K)$,
\[
\widetilde\omega(QYQ)=\Tr(\widetilde\omega\,QYQ)=\Tr(\widetilde\omega\,Y),
\]
and hence, when $Y\in Q\mathcal B(\mathcal K)Q$, the functional $\widetilde\omega$ agrees with the original
$\omega$ under the identification $Q\mathcal B(\mathcal K)Q\cong\mathcal B(\ran Q)$.

Applying this to $Y=\Phi_i(X)$ and $Y=\Phi(X)$ gives
\[
\omega\bigl((\Phi_i)_{Q,P}(X)\bigr)
=\omega\bigl(Q\,\Phi_i(X)\,Q\bigr)
=\widetilde\omega\bigl(\Phi_i(X)\bigr),
\qquad
\omega\bigl(\Phi_{Q,P}(X)\bigr)
=\omega\bigl(Q\,\Phi(X)\,Q\bigr)
=\widetilde\omega\bigl(\Phi(X)\bigr).
\]
Since $\Phi_i\to\Phi$ point-ultraweakly, we have
$\widetilde\omega(\Phi_i(X))\to \widetilde\omega(\Phi(X))$, and therefore
\[
\omega\!\bigl((\Phi_i)_{Q,P}(X)\bigr)
\longrightarrow
\omega\!\bigl(\Phi_{Q,P}(X)\bigr).
\]

Thus $(\Phi_i)_{Q,P}(X)\to \Phi_{Q,P}(X)$ ultraweakly in $Q\mathcal B(\mathcal K)Q\cong \mathcal B(\ran Q)$,
which is finite-dimensional; hence this convergence is in norm.

Choose a basis $(E_\alpha)_{\alpha=1}^{m^2}$ of $\mathcal B(\ran P)$ (where $m=\dim\ran P$). Since
$\|(\Phi_i)_{Q,P}(E_\alpha)-\Phi_{Q,P}(E_\alpha)\|\to 0$ for each $\alpha$ and $\mathcal L_{Q,P}$ is finite-dimensional,
it follows that $(\Phi_i)_{Q,P}\to \Phi_{Q,P}$ in the operator norm of $\mathcal L_{Q,P}$.
Because each $(\Phi_i)_{Q,P}\in\mathcal C_{Q,P}$ and $\mathcal C_{Q,P}$ is closed (C1), we obtain $\Phi_{Q,P}\in\mathcal C_{Q,P}$.
As $(Q,P)$ was arbitrary, $\Phi\in\mathsf K(\mathcal C)$, so $\mathsf K(\mathcal C)$ is point-ultraweakly closed.

\medskip
\noindent{(B) We show the exact recovery of corners: $(\mathsf K(\mathcal C))_{Q,P}=\mathcal C_{Q,P}$.}
\smallskip

Fix $(Q,P)\in\mathcal D$.
The inclusion $(\mathsf K(\mathcal C))_{Q,P}\subseteq \mathcal C_{Q,P}$ holds by definition of $\mathsf K(\mathcal C)$.
For the reverse inclusion, take $\phi\in\mathcal C_{Q,P}$ and define a normal map
$\widehat\phi:\mathcal B(\mathcal H)\to\mathcal B(\mathcal K)$ by the extension-by-zero
\begin{align}\label{eq:hatphi-final}
\widehat\phi(T):=Q\,\phi(PTP)\,Q,\qquad T\in\mathcal B(\mathcal H).
\end{align}

First note that $(\widehat\phi)_{Q,P}=\phi$.
For $X\in\mathcal B(\ran P)$ we have $PXP=X$, so
\[
(\widehat\phi)_{Q,P}(X)=Q\,\widehat\phi(X)\,Q=Q\,(Q\,\phi(X)\,Q)\,Q=\phi(X).
\]
Second,  we claim $\widehat\phi\in\mathsf K(\mathcal C)$.   Fix an arbitrary $(R,S)\in\mathcal D$.
Choose finite-rank $Q''\ge Q,R$ and $P''\ge P,S$.
For $X\in\mathcal B(\ran S)\subseteq\mathcal B(\ran P'')$ we compute
\begin{align*}
(\widehat\phi)_{R,S}(X)
&=R\,\widehat\phi(X)\,R
=R\,Q\,\phi(PXP)\,Q\,R\\
&=R\,Q''\,\phi(PXP)\,Q''\,R
=\Bigl[\rho_{(Q'',P'')\to(R,S)}\bigl(\iota_{(Q,P)\to(Q'',P'')}(\phi)\bigr)\Bigr](X),
\end{align*}
where we used $Q\le Q''$ and $R\le Q''$ to insert $Q''$ without changing the value.
By (C3), $\iota_{(Q,P)\to(Q'',P'')}(\phi)\in\mathcal C_{Q'',P''}$, and by (C2) its restriction to $(R,S)$ lies in $\mathcal C_{R,S}$.
Hence $(\widehat\phi)_{R,S}\in\mathcal C_{R,S}$ for all $(R,S)$, so $\widehat\phi\in\mathsf K(\mathcal C)$.
Therefore $\phi=(\widehat\phi)_{Q,P}\in (\mathsf K(\mathcal C))_{Q,P}$, proving $(\mathsf K(\mathcal C))_{Q,P}=\mathcal C_{Q,P}$.

\medskip
\noindent{(C) Classification and converse direction.}
\smallskip

Let $\mathsf K\subseteq \mathcal N(\mathcal B(\mathcal H),\mathcal B(\mathcal K))$ be a convex cut--pad stable point-ultraweakly
closed cone and define $\mathcal C^{\mathsf K}:=\{\mathsf K_{Q,P}\}_{(Q,P)\in\mathcal D}$.

\emph{Coherence of $\mathcal C^{\mathsf K}$.}
Convexity of each $\mathsf K_{Q,P}$ is immediate.

\smallskip
\emph{(C1) Closedness of each $\mathsf K_{Q,P}$.}
For finite-rank $(Q,P)$ and $\psi\in\mathcal L_{Q,P}$ define its extension-by-zero
\[
\widehat\psi(T):=Q\,\psi(PTP)\,Q,\qquad T\in\mathcal B(\mathcal H).
\]
We claim the equivalence
\begin{align}\label{eq:corner-iff-extension-final}
\psi\in \mathsf K_{Q,P}
\quad\Longleftrightarrow\quad
\widehat\psi\in\mathsf K.
\end{align}
Indeed, if $\psi\in\mathsf K_{Q,P}$ then $\psi=\Phi_{Q,P}$ for some $\Phi\in\mathsf K$. By heredity,
$\Phi^{(Q,P)}\in\mathsf K$, and $\Phi^{(Q,P)}=\widehat\psi$ by direct inspection.
Conversely, if $\widehat\psi\in\mathsf K$, then $(\widehat\psi)_{Q,P}=\psi$, so $\psi\in\mathsf K_{Q,P}$.

Now let $\psi_i\to\psi$ in $\mathcal L_{Q,P}$.
For fixed $T\in\mathcal B(\mathcal H)$ and $\omega\in\mathcal T(\mathcal K)$,
\[
\omega\!\bigl(\widehat\psi_i(T)-\widehat\psi(T)\bigr)
=\omega\!\bigl(Q\,(\psi_i-\psi)(PTP)\,Q\bigr)\longrightarrow 0,
\]
because $(\psi_i-\psi)(PTP)\to 0$ in norm in the finite-dimensional space $\mathcal B(\ran Q)$.
Hence $\widehat\psi_i\to\widehat\psi$ point-ultraweakly. Since $\mathsf K$ is point-ultraweakly closed,
$\widehat\psi\in\mathsf K$ whenever each $\widehat\psi_i\in\mathsf K$.
By \eqref{eq:corner-iff-extension-final}, this shows $\mathsf K_{Q,P}$ is closed in $\mathcal L_{Q,P}$.

\smallskip
\emph{(C2) Restriction coherence.}
If $(Q,P)\preceq(Q',P')$ and
$\psi\in\mathsf K_{Q',P'}$, write
$\psi=\Phi_{Q',P'}$ for some $\Phi\in\mathsf K$. Then
\[
\rho_{(Q',P')\to(Q,P)}(\psi)
=
\Phi_{Q,P}
\in\mathsf K_{Q,P},
\]
and hence
\[
\rho_{(Q',P')\to(Q,P)}
\bigl(\mathsf K_{Q',P'}\bigr)
\subseteq
\mathsf K_{Q,P}.
\]

\smallskip
\emph{(C3) Padding coherence.}
If $(Q,P)\preceq(Q',P')$ and $\phi\in\mathsf K_{Q,P}$, write $\phi=\Phi_{Q,P}$ with $\Phi\in\mathsf K$.
By heredity, $\Phi^{(Q,P)}\in\mathsf K$, and for $T\in\mathcal B(\ran P')$,
\begin{align*}
\bigl(\Phi^{(Q,P)}\bigr)_{Q',P'}(T)
&=Q'\,\Phi^{(Q,P)}(T)\,Q'
=Q'\,Q\,\Phi(PTP)\,Q\,Q'\\
&=Q'\,\phi(PTP)\,Q'
=\iota_{(Q,P)\to(Q',P')}(\phi)(T),
\end{align*}
using $Q\le Q'$.
Hence $\iota_{(Q,P)\to(Q',P')}(\phi)\in\mathsf K_{Q',P'}$.
Thus $\mathcal C^{\mathsf K}$ is coherent.

\medskip
\emph{Equality $\mathsf K=\mathsf K(\mathcal C^{\mathsf K})$.}
The inclusion $\mathsf K\subseteq \mathsf K(\mathcal C^{\mathsf K})$ is immediate.
For the reverse inclusion, let $\Phi\in \mathsf K(\mathcal C^{\mathsf K})$.
Then for each $(Q,P)\in\mathcal D$ we have $\Phi_{Q,P}\in \mathsf K_{Q,P}$, hence there exists $\Psi\in\mathsf K$
with $\Psi_{Q,P}=\Phi_{Q,P}$.
By heredity, $\Psi^{(Q,P)}\in\mathsf K$.
Moreover, for every $T\in\mathcal B(\mathcal H)$ we have $PTP\in P\mathcal B(\mathcal H)P$, hence
\[
\Psi^{(Q,P)}(T)=Q\,\Psi(PTP)\,Q
=Q\,\Phi(PTP)\,Q
=\Phi^{(Q,P)}(T),
\]
because $\Psi_{Q,P}=\Phi_{Q,P}$ means $Q\,\Psi(X)\,Q=Q\,\Phi(X)\,Q$ for all $X\in P\mathcal B(\mathcal H)P$.
Therefore $\Phi^{(Q,P)}\in\mathsf K$ for all $(Q,P)\in\mathcal D$.

Now we show that as $(Q,P)\in\mathcal D$ increases (i.e.\ $P\to I_{\mathcal H}$ and $Q\to I_{\mathcal K}$ strongly),
\begin{align}\label{eq:uw-approx-final}
\Phi^{(Q,P)} \longrightarrow \Phi \quad\text{in the point-ultraweak topology.}
\end{align}

Fix $T\in\mathcal B(\mathcal H)$ and $\omega\in\mathcal T(\mathcal K)$.
Using trace duality and Lemma~\ref{lem:trace-approx} for projections on $\mathcal H$, we have for any $S\in\mathcal T(\mathcal H)$,
\[
\Tr\!\bigl(S\,PTP\bigr)=\Tr\!\bigl(PSP\,T\bigr)\longrightarrow \Tr(ST),
\]
so $PTP\to T$ ultraweakly in $\mathcal B(\mathcal H)$.
Since $\Phi$ is normal, $\Phi(PTP)\to \Phi(T)$ ultraweakly in $\mathcal B(\mathcal K)$.
Moreover, Lemma~\ref{lem:trace-approx}, applied to \(Q\) on \(\mathcal K\), gives
\[
\|Q\omega Q-\omega\|_1\longrightarrow 0.
\]
Hence
\[
\begin{aligned}
&\left|
\omega\!\bigl(\Phi^{(Q,P)}(T)\bigr)
-\omega\!\bigl(\Phi(T)\bigr)
\right|\\
&\quad=
\left|
\Tr\!\bigl(Q\omega Q\,\Phi(PTP)\bigr)
-\Tr\!\bigl(\omega\Phi(T)\bigr)
\right|\\
&\quad\leq
\left|
\Tr\!\bigl((Q\omega Q-\omega)\Phi(PTP)\bigr)
\right|
+
\left|
\Tr\!\bigl(\omega(\Phi(PTP)-\Phi(T))\bigr)
\right|\\
&\quad\leq
\|Q\omega Q-\omega\|_1\,\|\Phi\|\,\|T\|
+
\left|
\Tr\!\bigl(\omega(\Phi(PTP)-\Phi(T))\bigr)
\right|
\longrightarrow 0.
\end{aligned}
\]
Indeed, the first term tends to zero by Lemma~\ref{lem:trace-approx}, while the second tends to zero because
\(\Phi(PTP)\to\Phi(T)\) ultraweakly. Thus
\[
\Phi^{(Q,P)}\longrightarrow\Phi
\quad\text{in the point-ultraweak topology},
\]
which proves \eqref{eq:uw-approx-final}.

\medskip
Since each $\Phi^{(Q,P)}\in\mathsf K$ and $\mathsf K$ is point-ultraweakly closed, \eqref{eq:uw-approx-final} implies $\Phi\in\mathsf K$.
Hence $\mathsf K(\mathcal C^{\mathsf K})\subseteq \mathsf K$, proving $\mathsf K=\mathsf K(\mathcal C^{\mathsf K})$.

Combining (A), (B), and the last equality shows that the assignments
$\mathcal C\mapsto \mathsf K(\mathcal C)$ and $\mathsf K\mapsto \{\mathsf K_{Q,P}\}$ are inverse bijections.
\end{proof}

The previous theorem relied on the assumption that $\mathsf K$ is point-ultraweakly closed. It is natural to ask how the reconstruction process behaves without this hypothesis. The following result shows that, for an arbitrary cut--pad stable cone, reconstruction from the norm closures of its finite-dimensional corner cones yields exactly the point-ultraweak closure of the original cone.

\begin{theorem}\label{thm:corner-closure}
Let $\mathsf K\subseteq \mathcal N(\mathcal B(\mathcal H),\mathcal B(\mathcal K))$ be a convex cut--pad stable cone (not necessarily point-ultraweakly closed).
Define the family of closed corner cones $\overline{\mathcal C}=\{\overline{\mathcal C}_{Q,P}\}_{(Q,P)\in\mathcal D}$ by taking the norm closure of the corners of $\mathsf K$:
\[
\overline{\mathcal C}_{Q,P} := \overline{(\mathsf K_{Q,P})}^{\|\cdot\|} \subseteq \mathcal L_{Q,P}.
\]
Then $\overline{\mathcal C}$ is a coherent family, and the reconstructed infinite cone is exactly the point-ultraweak closure of $\mathsf K$:
\[
\mathsf K(\overline{\mathcal C}) = \overline{\mathsf K}^{\,\mathrm{puw}}.
\]
\end{theorem}

\begin{proof}
First, we verify that $\overline{\mathcal C}$ satisfies the coherence conditions of Definition~\ref{def:coherent-family}:
\begin{enumerate}
\item[(C1)] Since $\mathsf K_{Q,P}$ is a convex cone, its norm closure $\overline{\mathcal C}_{Q,P}$ is a closed convex cone.
\item[(C2)] Let $(Q,P)\preceq (Q',P')$. The restriction map
$\rho=\rho_{(Q',P')\to(Q,P)}$ is a linear map between
finite-dimensional spaces, hence continuous. By the definition of the
corner cones,
\[
\rho(\mathsf K_{Q',P'})\subseteq \mathsf K_{Q,P}.
\]
Therefore,
\[
\rho(\overline{\mathcal C}_{Q',P'})
=
\rho\bigl(\overline{\mathsf K_{Q',P'}}\bigr)
\subseteq
\overline{\rho(\mathsf K_{Q',P'})}
\subseteq
\overline{\mathsf K_{Q,P}}
=
\overline{\mathcal C}_{Q,P}.
\]
\item[(C3)] Similarly, the padding map
$\iota=\iota_{(Q,P)\to(Q',P')}$ is continuous. Cut--pad stability implies
\[
\iota(\mathsf K_{Q,P})\subseteq \mathsf K_{Q',P'}.
\]
Hence,
\[
\iota(\overline{\mathcal C}_{Q,P})
=
\iota\bigl(\overline{\mathsf K_{Q,P}}\bigr)
\subseteq
\overline{\iota(\mathsf K_{Q,P})}
\subseteq
\overline{\mathsf K_{Q',P'}}
=
\overline{\mathcal C}_{Q',P'}.
\]
\end{enumerate}
Thus $\overline{\mathcal C}$ is coherent, and by Theorem~\ref{thm:Reconstruction}, the reconstructed cone $\mathsf K(\overline{\mathcal C})$ is point-ultraweakly closed.

We now prove the equality $\mathsf K(\overline{\mathcal C}) = \overline{\mathsf K}^{\,\mathrm{puw}}$ by showing two inclusions.

\noindent{Inclusion $\overline{\mathsf K}^{\,\mathrm{puw}} \subseteq \mathsf K(\overline{\mathcal C})$.}
Let $\Phi \in \mathsf K$. For any $(Q,P)$, the compression $\Phi_{Q,P}$ lies in $\mathsf K_{Q,P}$, which is contained in $\overline{\mathcal C}_{Q,P}$. Therefore $\Phi \in \mathsf K(\overline{\mathcal C})$ by definition.
Since $\mathsf K(\overline{\mathcal C})$ is point-ultraweakly closed and contains $\mathsf K$, it must contain the point-ultraweak closure $\overline{\mathsf K}^{\,\mathrm{puw}}$.

\noindent{Inclusion $\mathsf K(\overline{\mathcal C}) \subseteq \overline{\mathsf K}^{\,\mathrm{puw}}$.}
Let $\Phi \in \mathsf K(\overline{\mathcal C})$. By definition, for every $(Q,P) \in \mathcal D$,
\[
\Phi_{Q,P} \in \overline{\mathcal C}_{Q,P} = \overline{\mathsf K_{Q,P}}.
\]
Consider the net of cut--pad approximations defined by $\Phi^{(Q,P)}(T) := Q\,\Phi(PTP)\,Q$. As established in the proof of Theorem~\ref{thm:Reconstruction} (equation \eqref{eq:uw-approx-final}), $\Phi^{(Q,P)} \to \Phi$ in the point-ultraweak topology. Since $\overline{\mathsf K}^{\,\mathrm{puw}}$ is closed, it suffices to show that each term $\Phi^{(Q,P)}$ belongs to $\overline{\mathsf K}^{\,\mathrm{puw}}$.

Fix $(Q,P)$. Define the extension-by-zero map $\widehat{\psi}(T) := Q\,\psi(PTP)\,Q$ for $\psi \in \mathcal L_{Q,P}$. The image of this map is a finite-dimensional subspace of normal maps, where the norm topology and point-ultraweak topology coincide. Thus, the map $\psi \mapsto \widehat{\psi}$ is continuous. Note that $\Phi^{(Q,P)} = \widehat{\Phi_{Q,P}}$.

Since $\Phi_{Q,P} \in \overline{\mathsf K_{Q,P}}$, there exists a sequence $(\psi_n)_{n=1}^\infty \subseteq \mathsf K_{Q,P}$ such that $\psi_n \to \Phi_{Q,P}$ in norm.
By continuity,
\[
\widehat{\psi_n} \longrightarrow \widehat{\Phi_{Q,P}} = \Phi^{(Q,P)} \quad \text{point-ultraweakly.}
\]
Crucially, for each $n$, since $\psi_n \in \mathsf K_{Q,P}$, there exists (by definition of the corner cone) some global map $\Psi_n \in \mathsf K$ such that $(\Psi_n)_{Q,P} = \psi_n$. By the hereditary property of $\mathsf K$, the cut--pad map of this witness belongs to $\mathsf K$:
\[
\widehat{\psi_n} = \Psi_n^{(Q,P)} \in \mathsf K.
\]
Thus, $\Phi^{(Q,P)}$ is the ultraweak limit of the sequence $(\widehat{\psi_n})_n$ contained in $\mathsf K$, implying $\Phi^{(Q,P)} \in \overline{\mathsf K}^{\,\mathrm{puw}}$.
Finally, since $\Phi = \lim_{(Q,P)} \Phi^{(Q,P)}$, we conclude $\Phi \in \overline{\mathsf K}^{\,\mathrm{puw}}$.
\end{proof}

  \begin{corollary}\label{cor:sequential-reconstruction}
Suppose $\mathcal H$ and $\mathcal K$ are separable Hilbert spaces. Let $(P_n)_{n=1}^\infty$ and $(Q_n)_{n=1}^\infty$ be increasing sequences of finite-rank projections converging strongly to $I_{\mathcal H}$ and $I_{\mathcal K}$ respectively.
For any convex cut--pad stable cone $\mathsf K$ (closed or not), a normal map $\Phi$ belongs to the ultraweak closure $\overline{\mathsf K}^{\,\mathrm{puw}}$ if and only if
\[
\Phi_{Q_n, P_n} \in \overline{(\mathsf K_{Q_n, P_n})}^{\|\cdot\|} \quad \text{for all } n \in \mathbb{N}.
\]
\end{corollary}

\begin{proof}
If $\Phi \in \overline{\mathsf K}^{\,\mathrm{puw}}$, then for any $n$, continuity of the compression implies $\Phi_{Q_n, P_n} \in \overline{(\mathsf K_{Q_n, P_n})}^{\|\cdot\|}$.

Conversely, assume the condition holds for all $n$. Consider the sequence of cut--pad approximations $\Phi_n(T) := Q_n \Phi(P_n T P_n) Q_n$. As shown in Theorem~\ref{thm:Reconstruction}, $\Phi_n \to \Phi$ in the point-ultraweak topology.
We claim each $\Phi_n$ belongs to $\overline{\mathsf K}^{\,\mathrm{puw}}$.
Note that $\Phi_n$ is the extension-by-zero of the corner map $\psi_n := \Phi_{Q_n, P_n}$.
By hypothesis, $\psi_n$ is in the norm closure of $\mathsf K_{Q_n, P_n}$. Hence, there exist maps $\psi_{n,k} \in \mathsf K_{Q_n, P_n}$ such that $\psi_{n,k} \to \psi_n$ in norm as $k\to\infty$.
By definition of the corner cone, each $\psi_{n,k}$ is the compression of some $\Psi_{n,k} \in \mathsf K$. By heredity, the cut--pad extensions satisfy $\Psi_{n,k}^{(Q_n, P_n)} \in \mathsf K$.
The extension map is continuous, so $\Psi_{n,k}^{(Q_n, P_n)} \to \Phi_n$ point-ultraweakly as $k\to\infty$. Thus $\Phi_n$ is a limit of elements in $\mathsf K$, so $\Phi_n \in \overline{\mathsf K}^{\,\mathrm{puw}}$.
Finally, since $\overline{\mathsf K}^{\,\mathrm{puw}}$ is closed and $\Phi_n \to \Phi$, we conclude $\Phi \in \overline{\mathsf K}^{\,\mathrm{puw}}$.
\end{proof}


\section{Corner Polar Duality}

The reconstruction results of Section~3 show that a
point-ultraweakly closed cut--pad stable cone
$\mathsf K\subseteq \mathcal N(\mathcal B(\mathcal H),\mathcal B(\mathcal K))$
is completely determined by its finite-dimensional corners
$\{\mathsf K_{Q,P}\}_{(Q,P)\in\mathcal D}$. In this section,  we derive a finite-corner separation criterion. For each corner we use the
ordinary finite-dimensional polar cone, expressed through the Choi--trace
pairing, and combine these local separation criteria with the reconstruction
theorem.
\medskip

For $(Q,P)\in\mathcal D$, $W\in\mathcal T(\ran P\otimes\ran Q)$, and
$\psi\in\mathcal L_{Q,P}$, we define
\[
\langle W,\psi\rangle_{Q,P}
:=
\Re\,\Tr(WC_\psi).
\]
For each $(Q,P)\in\mathcal D$ define the polar cone of $\mathsf K_{Q,P}$   by
\[
(\mathsf K_{Q,P})^\circ
:=
\Bigl\{
W\in\mathcal T(\ran P\otimes\ran Q)\;:\;
\langle W,\psi\rangle_{Q,P}\ge 0\ \ \forall\,\psi\in\mathsf K_{Q,P}
\Bigr\}.
\]
We view $(\mathsf K_{Q,P})^\circ$ as a cone in $\mathcal T(\mathcal H\otimes\mathcal K)$ by extension by zero outside
$\ran(P)\otimes\ran(Q)$, using the identification above.

Define the \emph{finite-corner polar family} of $\mathsf K$ by
\[
\mathsf K_{\mathrm{fin}}^\circ
:=
\bigsqcup_{(Q,P)\in\mathcal D}
\bigl\{(Q,P,W): W\in(\mathsf K_{Q,P})^\circ\bigr\}.
\]
Thus an element of $\mathsf K_{\mathrm{fin}}^\circ$ retains the finite
corner on which the corresponding witness is defined. We define its
corner bipolar by
\[
(\mathsf K_{\mathrm{fin}}^\circ)^\circ
:=
\Bigl\{
\Phi\in \mathcal N(\mathcal B(\mathcal H),\mathcal B(\mathcal K)):\
\langle W,\Phi_{Q,P}\rangle_{Q,P}\ge 0
\ \text{for every }(Q,P,W)\in\mathsf K_{\mathrm{fin}}^\circ
\Bigr\}.
\]

\begin{theorem}\label{thm:corner-bipolar}
Let $\mathsf K\subseteq \mathcal N(\mathcal B(\mathcal H),\mathcal B(\mathcal K))$ be a convex cut--pad stable cone
which is point-ultraweakly closed. Then
\[
\boxed{\qquad
\mathsf K
=
(\mathsf K_{\mathrm{fin}}^\circ)^\circ
\qquad}.
\]
Equivalently, for $\Phi\in \mathcal N(\mathcal B(\mathcal H),\mathcal B(\mathcal K))$ one has
\[
\Phi\in\mathsf K
\quad\Longleftrightarrow\quad
\langle W,\Phi_{Q,P}\rangle_{Q,P}\ge 0
\ \ \text{for all }(Q,P)\in\mathcal D\text{ and all }W\in(\mathsf K_{Q,P})^\circ.
\]
\end{theorem}

\begin{proof}
First note that since $\mathsf K$ is convex, cut--pad stable, and point-ultraweakly closed, the reconstruction Theorem~\ref{thm:Reconstruction}  yields
\begin{equation}\label{eq:corner-membership-proof-final}
\Phi\in\mathsf K
\quad\Longleftrightarrow\quad
\Phi_{Q,P}\in \mathsf K_{Q,P}\ \text{ for all }(Q,P)\in\mathcal D.
\end{equation}

Moreover, by Theorem~\ref{thm:Reconstruction}, the family
$\{\mathsf K_{Q,P}\}_{(Q,P)\in\mathcal D}$ is coherent; in particular,
each $\mathsf K_{Q,P}$ is a closed convex cone in the finite-dimensional
space $\mathcal L_{Q,P}$.

\noindent
We now prove the inclusion $\mathsf K\subseteq(\mathsf K_{\mathrm{fin}}^\circ)^\circ$.
Let $\Phi\in\mathsf K$. Fix $(Q,P)\in\mathcal D$ and $W\in(\mathsf K_{Q,P})^\circ$.
Since $\Phi_{Q,P}\in\mathsf K_{Q,P}$, the definition of $(\mathsf K_{Q,P})^\circ$ gives
\[
\langle W,\Phi_{Q,P}\rangle_{Q,P}\ge 0.
\]
As $(Q,P)$ and $W$ were arbitrary, $\Phi\in(\mathsf K_{\mathrm{fin}}^\circ)^\circ$.

 \noindent
To prove the reverse inclusion $\mathsf K\supseteq(\mathsf K_{\mathrm{fin}}^\circ)^\circ$ we proceed contrapositive.
Assume $\Phi\notin\mathsf K$. By \eqref{eq:corner-membership-proof-final} there exists $(Q,P)\in\mathcal D$ with
$\Phi_{Q,P}\notin \mathsf K_{Q,P}$.  As shown above, $\mathsf K_{Q,P}$ is a closed convex cone in the finite-dimensional
real vector space $(\mathcal L_{Q,P})_{\mathbb R}$. Hence there exists a nonzero real-linear functional
$\ell:(\mathcal L_{Q,P})_{\mathbb R}\to\mathbb R$ such that
\begin{equation}\label{eq:sep-functional-proof-final}
\ell(\psi)\ge 0\quad \forall\,\psi\in\mathsf K_{Q,P},
\qquad\text{and}\qquad
\ell(\Phi_{Q,P})<0.
\end{equation}

Let $J_{Q,P}:\mathcal L_{Q,P}\to \mathcal B(\ran P\otimes\ran Q)$ be the Choi isomorphism $J_{Q,P}(\psi)=C_\psi$,
and set $\tilde\ell:=\ell\circ J_{Q,P}^{-1}$. Equip $\mathcal B(\ran P\otimes\ran Q)$ with the real Hilbert--Schmidt
inner product $\langle X,Y\rangle_{\mathrm{HS},\mathbb R}:=\Re\,\Tr(X^*Y)$. By the real Riesz representation theorem,
there exists $A\in\mathcal B(\ran P\otimes\ran Q)$ such that
\[
\tilde\ell(Z)=\Re\,\Tr(A^*Z)\qquad \forall\,Z\in\mathcal B(\ran P\otimes\ran Q).
\]
Set $W_0:=A^*\in\mathcal T(\ran P\otimes\ran Q)$ (since $\mathcal B=\mathcal T$ in finite dimension). Then
\begin{equation}\label{eq:choi-rep-proof-final}
\ell(\psi)=\Re\,\Tr(W_0\,C_\psi)=\langle W_0,\psi\rangle_{Q,P}\qquad \forall\,\psi\in\mathcal L_{Q,P}.
\end{equation}
From \eqref{eq:sep-functional-proof-final} and \eqref{eq:choi-rep-proof-final} we obtain
\[
\langle W_0,\psi\rangle_{Q,P}\ge 0\ \ \forall\,\psi\in\mathsf K_{Q,P},
\qquad\text{and}\qquad
\langle W_0,\Phi_{Q,P}\rangle_{Q,P}<0.
\]
Thus $W_0\in(\mathsf K_{Q,P})^\circ$, so
\[
(Q,P,W_0)\in\mathsf K_{\mathrm{fin}}^\circ.
\]
Moreover,
\[
\langle W_0,\Phi_{Q,P}\rangle_{Q,P}<0.
\]
Hence, by the definition of the corner bipolar,
\[
\Phi\notin(\mathsf K_{\mathrm{fin}}^\circ)^\circ.
\]

\medskip\noindent
The two inclusions established above show that $\mathsf K=(\mathsf K_{\mathrm{fin}}^\circ)^\circ$.
The displayed equivalence in the statement follows immediately from the definition of
$(\mathsf K_{\mathrm{fin}}^\circ)^\circ$.
\end{proof}

 \begin{remark} 
Theorem~\ref{thm:corner-bipolar} combines the reconstruction theorem
with ordinary finite-dimensional separation on each corner. Its point is
that no global separating functional on the space of normal maps is needed:
if $\Phi\notin\mathsf K$, then non-membership is already detected on a
single finite-dimensional corner by a Choi--trace witness from the
corresponding polar cone.

Thus, for point-ultraweakly closed cut--pad stable cones, the tagged family
of finite-corner polar witnesses is sufficient both to detect
non-membership and to recover the original cone.
\end{remark}
\begin{corollary}\label{cor:finite-witness}
Let $\mathsf K$ be a convex cut--pad stable  point-ultraweakly closed cone. If $\Phi\notin\mathsf K$, then there exist
$(Q,P)\in\mathcal D$ and a nonzero $W\in(\mathsf K_{Q,P})^\circ\subseteq \mathcal T(\ran P\otimes\ran Q)$ such that
\[
\langle W,\Phi_{Q,P}\rangle_{Q,P}<0.
\]
Equivalently, the failure of $\Phi$ to belong to $\mathsf K$ is detected by a
witness on a single finite-dimensional corner.
\end{corollary}

\begin{proof}
Since $\Phi\notin\mathsf K=(\mathsf K_{\mathrm{fin}}^\circ)^\circ$ by Theorem~\ref{thm:corner-bipolar}, there exist
$(Q,P)\in\mathcal D$ and $W\in(\mathsf K_{Q,P})^\circ$ such that
$\langle W,\Phi_{Q,P}\rangle_{Q,P}<0$.
\end{proof}


\section{Examples and Discussion}
We apply the preceding results to completely positive and decomposable maps. The completely
positive case recovers the finite-corner form of Choi positivity, while the
decomposable case leads to finite-corner PPT witnesses and is shown to agree
with St{\o}rmer's classical notion of decomposability. We conclude with a
finite-rank completely positive example showing that the point-ultraweak
closedness assumption in the reconstruction theorem is essential.

\begin{example}\label{ex:CP}
Let $\mathsf{CP}\subseteq \mathcal N(\mathcal B(\mathcal H),\mathcal B(\mathcal K))$ denote the cone of
\emph{normal completely positive} maps. This cone is convex, cut--pad stable, and point-ultraweakly closed.

Fix $(Q,P)\in\mathcal D$. Then $\mathsf{CP}_{Q,P}\subseteq \mathcal L_{Q,P}$ consists of completely positive maps
$\psi:\mathcal B(\ran P)\to\mathcal B(\ran Q)$. In finite dimension, $\psi$ is completely positive if and only if its
Choi operator is positive \cite{Choi1975,Jamiolkowski1972}:
\[
\mathsf{CP}_{Q,P}
=
\bigl\{\psi\in\mathcal L_{Q,P}:\ C_\psi\ge 0\bigr\}.
\]
(The positivity criterion is independent of the particular choice of matrix units.)

Moreover, with respect to the real pairing
$\langle W,\psi\rangle_{Q,P}=\Re\,\Tr(WC_\psi)$, one has
\[
(\mathsf{CP}_{Q,P})^\circ
=
\left\{
W\in\mathcal T(\ran P\otimes\ran Q):
\frac{W+W^*}{2}\ge 0
\right\}.
\]
Indeed, set
\[
H_W:=\frac{W+W^*}{2}.
\]
If $H_W\ge 0$ and $C_\psi\ge 0$, then
\[
\Re\,\Tr(WC_\psi)
=
\Tr(H_WC_\psi)
\ge 0.
\]
Conversely, if $H_W\not\ge 0$, choose $v\in\ran P\otimes\ran Q$ with
$\langle v,H_Wv\rangle<0$. By the Choi--Jamio\l{}kowski isomorphism, there exists
$\psi\in\mathsf{CP}_{Q,P}$ such that
\[
C_\psi=|v\rangle\langle v|.
\]
It follows that
\[
\Re\,\Tr(WC_\psi)
=
\Tr(H_WC_\psi)
=
\langle v,H_Wv\rangle
<0,
\]
so $W\notin(\mathsf{CP}_{Q,P})^\circ$.

Consequently,
\[
\mathsf{CP}_{\mathrm{fin}}^\circ
=
\bigsqcup_{(Q,P)\in\mathcal D}
\left\{
(Q,P,W):
W\in\mathcal T(\ran P\otimes\ran Q),\
\frac{W+W^*}{2}\ge0
\right\}.
\]
Thus a finite-corner polar witness for complete positivity consists of a
finite corner $(Q,P)$ together with an operator on that corner whose
Hermitian part is positive.

Applying Theorem~\ref{thm:corner-bipolar} to $\mathsf K=\mathsf{CP}$ yields the corner bipolar description
\[
\mathsf{CP}
=
(\mathsf{CP}_{\mathrm{fin}}^\circ)^\circ,
\]
i.e.\ a normal map $\Phi$ is completely positive if and only if
\[
\Re\,\Tr\!\bigl(W\,C_{\Phi_{Q,P}}\bigr)\ge 0
\]
for all $(Q,P)\in\mathcal D$ and all
$W\in\mathcal T(\ran P\otimes\ran Q)$ with
\[
\frac{W+W^*}{2}\ge 0.
\]
Indeed, the admissible witnesses include all skew-adjoint operators $W=iK$,
with $K=K^*$, as well as their negatives; hence the above inequalities force
$C_{\Phi_{Q,P}}$ to be self-adjoint. Testing then against all positive
self-adjoint $W$ forces $C_{\Phi_{Q,P}}\ge 0$. Thus, equivalently,
$\Phi\in\mathsf{CP}$ if and only if $C_{\Phi_{Q,P}}\ge 0$ for every finite
corner $(Q,P)$.
\end{example}

\begin{example}\label{ex:decomposable}
Fix $(Q,P)\in\mathcal D$ and use the fixed matrix units on $\mathcal B(\ran P)\cong M_m$ to define the
\emph{transpose} map $\tau_P:\mathcal B(\ran P)\to\mathcal B(\ran P)$.
A linear map $\psi\in\mathcal L_{Q,P}$ is called \emph{completely copositive} on the corner $(Q,P)$ if
$\psi\circ\tau_P$ is completely positive, and we denote the corresponding cone by
\[
\mathsf{coCP}_{Q,P}:=
\{\psi\in\mathcal L_{Q,P}:\ \psi\circ\tau_P\in\mathsf{CP}_{Q,P}\}.
\]
Equivalently, in terms of Choi operators,
\[
\psi\in\mathsf{coCP}_{Q,P}
\quad\Longleftrightarrow\quad
(\tau_P\otimes\mathrm{id})(C_\psi)\ge 0,
\]
where $(\tau_P\otimes\mathrm{id})$ is the partial transpose on
$\mathcal B(\ran P\otimes\ran Q)$ with respect to the chosen matrix units on $\ran P$.

Define the \emph{decomposable corner cone} by
\[
\mathsf{Dec}_{Q,P}
:=
\mathsf{CP}_{Q,P}+\mathsf{coCP}_{Q,P}
\subseteq\mathcal L_{Q,P}.
\]
Thus $\psi\in\mathsf{Dec}_{Q,P}$ if and only if
$\psi=\psi_1+\psi_2$, where $\psi_1$ is completely positive and
$\psi_2$ is completely copositive on the corner.

The family $\{\mathsf{Dec}_{Q,P}\}_{(Q,P)\in\mathcal D}$ is coherent.
First, complete copositivity is independent of the   choice of
matrix units. Indeed, if $\tau_P'$ is another transpose map, then
\[
\tau_P'=\tau_P\circ\alpha
\]
for some $*$-automorphism $\alpha$ of $\mathcal B(\ran P)$. Hence
\[
\psi\circ\tau_P'
=
(\psi\circ\tau_P)\circ\alpha,
\]
so $\psi\circ\tau_P'$ is completely positive if and only if
$\psi\circ\tau_P$ is completely positive.

Now suppose $(Q,P)\preceq(Q',P')$. To verify restriction and padding
coherence, we can choose matrix units on $\ran P'$ extending the fixed
matrix units on $\ran P$. With respect to these matrix units,
\[
\tau_{P'}(X)=\tau_P(X),
\qquad X\in\mathcal B(\ran P),
\]
and
\[
P\tau_{P'}(T)P=\tau_P(PTP),
\qquad T\in\mathcal B(\ran P').
\]
Hence,  restriction and padding preserve both complete positivity
and complete copositivity. By the independence of the choice of matrix
units noted above, the same conclusion holds for the originally fixed
matrix units. Thus conditions {\rm (C2)} and {\rm (C3)} are satisfied.

It remains to verify condition {\rm (C1)}. The cone
$\mathsf{Dec}_{Q,P}$ is   convex. Set
\[
\Gamma_P:=\tau_P\otimes\mathrm{id}.
\]
In Choi form,
\[
\mathsf{Dec}_{Q,P}
=
\left\{
\psi\in\mathcal L_{Q,P}:
C_\psi=A+\Gamma_P(B)
\text{ for some }A,B\ge 0
\right\}.
\]
Suppose that
\[
A_n+\Gamma_P(B_n)\longrightarrow C,
\qquad A_n,B_n\ge 0.
\]
Since $\Gamma_P$ is trace preserving,
\[
\Tr(A_n)+\Tr(B_n)
=
\Tr\!\bigl(A_n+\Gamma_P(B_n)\bigr).
\]
The right-hand side is bounded, and hence the positive sequences
$(A_n)$ and $(B_n)$ are bounded. By finite dimensionality, after passing
to a subsequence,
\[
A_n\longrightarrow A,
\qquad
B_n\longrightarrow B
\]
for some $A,B\ge 0$. Therefore
\[
C=A+\Gamma_P(B),
\]
so $\mathsf{Dec}_{Q,P}$ is closed.

Hence, by the finite-to-infinite reconstruction
Theorem~\ref{thm:Reconstruction}, there exists a canonical convex cut--pad stable
point-ultraweakly closed cone of normal maps, denoted
\[
\mathsf{Dec}
\subseteq
\mathcal N(\mathcal B(\mathcal H),\mathcal B(\mathcal K)),
\]
whose corner cones are exactly $\mathsf{Dec}_{Q,P}$:
\[
\mathsf{Dec}_{Q,P}=(\mathsf{Dec})_{Q,P}
\qquad
\forall (Q,P)\in\mathcal D.
\]
In finite dimensions, when $P=I_{\mathcal H}$ and $Q=I_{\mathcal K}$,
this recovers the usual cone of decomposable maps.

\smallskip
We now compute the corner polars. For
$W\in\mathcal T(\ran P\otimes\ran Q)$, set
\[
H_W:=\frac{W+W^*}{2}.
\]
On the corner $(Q,P)$, Example~\ref{ex:CP} gives
\[
(\mathsf{CP}_{Q,P})^\circ
=
\left\{
W\in\mathcal T(\ran P\otimes\ran Q):
H_W\ge 0
\right\}.
\]

For the copositive cone, using that $\Gamma_P$ is self-adjoint with respect
to the trace pairing in finite dimensions,
\[
\Tr\!\bigl(W\,\Gamma_P(Z)\bigr)
=
\Tr\!\bigl(\Gamma_P(W)\,Z\bigr)
\qquad
\bigl(
W\in\mathcal T(\ran P\otimes\ran Q),\
Z\in\mathcal B(\ran P\otimes\ran Q)
\bigr),
\]
we obtain
\[
(\mathsf{coCP}_{Q,P})^\circ
=
\left\{
W\in\mathcal T(\ran P\otimes\ran Q):
\Gamma_P(H_W)\ge 0
\right\}.
\]
Indeed, $\psi\in\mathsf{coCP}_{Q,P}$ if and only if
$\Gamma_P(C_\psi)\ge 0$. Since $\Gamma_P$ is involutive, this means that
\[
C_\psi=\Gamma_P(Z)
\qquad
\text{for some }Z\ge 0.
\]
Since $\Gamma_P$ commutes with taking adjoints, it follows that
\[
\begin{aligned}
\Re\,\Tr(WC_\psi)
&=
\Re\,\Tr\!\bigl(W\Gamma_P(Z)\bigr)\\
&=
\Re\,\Tr\!\bigl(\Gamma_P(W)Z\bigr)\\
&=
\Tr\!\bigl(\Gamma_P(H_W)Z\bigr).
\end{aligned}
\]
By the self-duality of the positive cone, this expression is nonnegative
for every $Z\ge 0$ if and only if
\[
\Gamma_P(H_W)\ge 0.
\]

Since polars satisfy
\[
(\mathsf{CP}_{Q,P}+\mathsf{coCP}_{Q,P})^\circ
=
(\mathsf{CP}_{Q,P})^\circ
\cap
(\mathsf{coCP}_{Q,P})^\circ,
\]
it follows that
\[
(\mathsf{Dec}_{Q,P})^\circ
=
\left\{
W\in\mathcal T(\ran P\otimes\ran Q):
H_W\ge 0
\ \text{and}\
\Gamma_P(H_W)\ge 0
\right\}.
\]
Thus the Hermitian parts of the corner witnesses for decomposability are
  the \emph{PPT} operators on the corner, that is, the positive
operators which remain positive under partial transpose; see
\cite{Horodecki1996,Terhal2000}.

Consequently, the finite-corner polar family is
\[
\mathsf{Dec}_{\mathrm{fin}}^\circ
=
\bigsqcup_{(Q,P)\in\mathcal D}
\left\{
(Q,P,W):
W\in\mathcal T(\ran P\otimes\ran Q),\
H_W\ge0,\ \Gamma_P(H_W)\ge0
\right\}.
\]
Thus a finite-corner polar witness for decomposability consists of a
finite corner $(Q,P)$ together with an operator whose Hermitian part is
PPT on that corner. Theorem~\ref{thm:corner-bipolar} therefore yields the bipolar
description
\[
\mathsf{Dec}
=
(\mathsf{Dec}_{\mathrm{fin}}^\circ)^\circ,
\]
i.e.\ $\Phi\in\mathsf{Dec}$ if and only if
\[
\langle W,\Phi_{Q,P}\rangle_{Q,P}\ge 0
\]
for every $(Q,P)\in\mathcal D$ and every
$W\in\mathcal T(\ran P\otimes\ran Q)$ satisfying
\[
\frac{W+W^*}{2}\ge 0
\qquad\text{and}\qquad
\Gamma_P\!\left(\frac{W+W^*}{2}\right)\ge 0.
\]
In particular, if $\Phi\notin\mathsf{Dec}$, then this failure is detected
on a single finite corner by a witness whose Hermitian part is PPT.
\end{example}

\begin{remark}\label{rem:classical-decomposability}
The cone $\mathsf{Dec}$ constructed above coincides with the cone of
normal decomposable maps in the sense of St{\o}rmer. Recall that
St{\o}rmer calls a positive map
\[
\Phi:\mathcal A\longrightarrow\mathcal B(\mathcal K)
\]
decomposable if there exist a Hilbert space $\mathcal L$, a bounded
operator $V:\mathcal K\to\mathcal L$, and a Jordan $*$-homomorphism
\[
\pi:\mathcal A\longrightarrow\mathcal B(\mathcal L)
\]
such that
\[
\Phi(X)=V^*\pi(X)V
\qquad (X\in\mathcal A);
\]
see \cite[Definition~7.1]{Stormer1963}. In the terminology of
\cite{Stormer1963}, $\pi$ is called a $C^*$-homomorphism. St{\o}rmer
proved in \cite{Stormer1982} that this is equivalent to the following
matrix condition: for every $n\in\mathbb N$,
\[
[X_{ij}]\geq0
\quad\text{and}\quad
[X_{ji}]\geq0
\]
in $M_n(\mathcal A)$ imply
\[
[\Phi(X_{ij})]\geq0
\]
in $M_n(\mathcal B(\mathcal K))$.

For maps between full matrix algebras, St{\o}rmer decomposability is
equivalent to being a sum of a completely positive map and a completely
copositive map. Indeed, a Jordan $*$-homomorphism decomposes into a
$*$-homomorphic part and a $*$-antihomomorphic part
\cite[Theorem~3.3]{Stormer1965}. Compression of the former
gives a completely positive map, while, after composition with a
transpose on the domain, compression of the latter gives a completely
positive map. The converse follows by taking the direct sum of
Stinespring representations for the completely positive part and for
the completely copositive part after composition with the transpose.

We now prove the  coincidence. Suppose first that
$\Phi\in\mathsf{Dec}$, and let
\[
[X_{ij}]\geq0
\quad\text{and}\quad
[X_{ji}]\geq0
\]
in $M_n(\mathcal B(\mathcal H))$. For finite-rank projections
$P\in\mathcal B(\mathcal H)$ and $Q\in\mathcal B(\mathcal K)$, we have
\[
[PX_{ij}P]\geq0
\quad\text{and}\quad
[PX_{ji}P]\geq0.
\]
Since
\[
\Phi_{Q,P}\in\mathsf{Dec}_{Q,P},
\]
the finite-dimensional matrix criterion gives
\[
[Q\Phi(PX_{ij}P)Q]\geq0.
\]
Letting $P\uparrow I_{\mathcal H}$ through the directed set of
finite-rank projections and using the normality of $\Phi$, together
with the ultraweak closedness of the positive cone, we obtain
\[
[Q\Phi(X_{ij})Q]\geq0.
\]
This holds for every finite-rank projection $Q$. Given arbitrary
$\xi_1,\ldots,\xi_n\in\mathcal K$, choose $Q$ whose range contains
$\xi_1,\ldots,\xi_n$. It follows that
\[
\sum_{i,j=1}^{n}
\bigl\langle\Phi(X_{ij})\xi_j,\xi_i\bigr\rangle
\geq0.
\]
Hence
\[
[\Phi(X_{ij})]\geq0.
\]
St{\o}rmer's matrix characterization  shows that $\Phi$ is
decomposable.

Conversely, suppose that $\Phi$ is a normal decomposable map in the
sense of St{\o}rmer, and fix $(Q,P)\in\mathcal D$. If
\[
[X_{ij}]\geq0
\quad\text{and}\quad
[X_{ji}]\geq0
\]
in $M_n(\mathcal B(\ran P))$, then, after identifying
$\mathcal B(\ran P)$ with $P\mathcal B(\mathcal H)P$, St{\o}rmer's
matrix characterization gives
\[
[\Phi(X_{ij})]\geq0.
\]
Compression by $I_n\otimes Q$ yields
\[
[\Phi_{Q,P}(X_{ij})]
=
[Q\Phi(X_{ij})Q]
\geq0.
\]
Thus $\Phi_{Q,P}$ is decomposable in the sense of St{\o}rmer. By the
finite-dimensional equivalence noted above,
\[
\Phi_{Q,P}
\in
\mathsf{CP}_{Q,P}+\mathsf{coCP}_{Q,P}
=
\mathsf{Dec}_{Q,P}.
\]
Since this holds for every $(Q,P)\in\mathcal D$, the definition of the
reconstructed cone gives $\Phi\in\mathsf{Dec}$.

Consequently,
\[
\mathsf{Dec}
=
\left\{
\Phi\in
\mathcal N\bigl(\mathcal B(\mathcal H),\mathcal B(\mathcal K)\bigr):
\Phi\text{ is decomposable in the sense of St{\o}rmer}
\right\}.
\]
Decomposable positive maps in low-dimensional matrix algebras were
studied by Woronowicz; in particular, he proved that every positive
map from $\mathbb{M}_2$ into $\mathbb{M}_3$ is decomposable \cite{Woronowicz1976}.
For decomposable positive projections on $C^*$-algebras, see
\cite{Stormer1980}.
\end{remark}

\subsection{Sharpness of assumptions and comparison with classical duality}\label{subsec:sharpness}

\begin{remark}\label{rem:automatic-vs-nontrivial}
By definition,
\[
\mathsf K_{Q,P}:=\{\Psi_{Q,P}:\Psi\in\mathsf K\}.
\]
Therefore, if $\Phi\in\mathsf K$ then its compression $\Phi_{Q,P}$ automatically belongs to $\mathsf K_{Q,P}$ for every
finite corner $(Q,P)\in\mathcal D$.

The converse contains the local-to-global content of the reconstruction
theorem.  The condition $\Phi_{Q,P}\in\mathsf K_{Q,P}$ means that for each
corner $(Q,P)$ one can find, possibly depending on $(Q,P)$, a map
$\Psi^{Q,P}\in\mathsf K$ whose compression agrees with $\Phi$ on that
corner:
\[
(\Psi^{Q,P})_{Q,P}=\Phi_{Q,P}.
\]
These representatives need not arise from a single element of $\mathsf K$.
Cut--pad stability replaces each of them by the corresponding cut--pad map,
while point-ultraweak closedness allows the resulting net of local
approximants to recover global membership.
\end{remark}

We next give a simple example showing that the point-ultraweak closedness
assumption in our reconstruction and bipolar results is essential: without
it, a cut--pad stable cone may have the \emph{same} finite-corner cones as
its point-ultraweak closure, so cornerwise membership tests (and hence
finite-corner duality) need not characterize membership in the original
cone.

\begin{proposition}\label{prop:need-uw-closed}
Assume $\mathcal H=\mathcal K$ is infinite-dimensional. Let $\mathsf K_{\mathrm{fr}}$ be the cone of all normal
completely positive maps $\Phi:\mathcal B(\mathcal H)\to\mathcal B(\mathcal H)$ whose range is contained in a
finite-dimensional subspace of $\mathcal B(\mathcal H)$ (equivalently, $\Phi$ has finite rank as a linear map).
Then:
\begin{enumerate}
\item[(a)] $\mathsf K_{\mathrm{fr}}$ is a convex cut--pad stable cone of normal maps, but it is \emph{not} point-ultraweakly closed.
\item[(b)] For every finite corner $(Q,P)\in\mathcal D$ one has
\[
(\mathsf K_{\mathrm{fr}})_{Q,P}=\mathsf{CP}_{Q,P}.
\]
\item[(c)] The identity map $\mathrm{id}_{\mathcal B(\mathcal H)}$ satisfies
\[
(\mathrm{id})_{Q,P}\in (\mathsf K_{\mathrm{fr}})_{Q,P}\ \ \forall(Q,P)\in\mathcal D,
\qquad\text{but}\qquad
\mathrm{id}\notin \mathsf K_{\mathrm{fr}}.
\]
\end{enumerate}
Consequently, the implication
\[
\bigl(\forall(Q,P)\in\mathcal D,\ \Phi_{Q,P}\in \mathsf K_{Q,P}\bigr)\ \Longrightarrow\ \Phi\in\mathsf K
\]
can fail if $\mathsf K$ is not assumed point-ultraweakly closed.
\end{proposition}

\begin{proof}
(a) Convexity is clear. If $\Phi$ has finite-dimensional range, then for any finite-rank $P,Q$ the cut--pad map
$\Phi^{(Q,P)}(T):=Q\,\Phi(PTP)\,Q$ has range contained in $Q\,\mathrm{Ran}(\Phi)\,Q$, hence still finite-dimensional;
moreover $\Phi^{(Q,P)}$ is normal and completely positive. Thus $\mathsf K_{\mathrm{fr}}$ is hereditary.

To see that $\mathsf K_{\mathrm{fr}}$ is not point-ultraweakly closed, consider the directed set of finite-rank
projections $R$ on $\mathcal H$ ordered by inclusion and define
\[
\Phi_R(T):=RTR,\qquad T\in\mathcal B(\mathcal H).
\]
Each $\Phi_R$ is normal and completely positive, and $\mathrm{Ran}(\Phi_R)\subseteq R\mathcal B(\mathcal H)R
\cong \mathcal B(\ran R)$ is finite-dimensional; hence $\Phi_R\in\mathsf K_{\mathrm{fr}}$.
We claim $\Phi_R\to \mathrm{id}$ point-ultraweakly. Indeed, fix $T\in\mathcal B(\mathcal H)$ and a normal functional
$\omega_S(\cdot)=\Tr(S\,\cdot)$ with $S\in\mathcal T(\mathcal H)$. Then
\[
\omega_S(\Phi_R(T))=\Tr(SRTR)=\Tr(RSR\,T).
\]
Since $R\to I$ strongly and $S$ is trace-class, it follows from Lemma ~\ref{lem:trace-approx} that $\|RSR-S\|_1\to 0$; therefore
$\Tr(RSR\,T)\to \Tr(S\,T)=\omega_S(T)$.
Hence $\Phi_R\to\mathrm{id}$ in the point-ultraweak topology on normal maps.
But $\mathrm{id}$ does not have finite-dimensional range, so $\mathrm{id}\notin\mathsf K_{\mathrm{fr}}$.
Thus $\mathsf K_{\mathrm{fr}}$ is not point-ultraweakly closed.

(b) Fix $(Q,P)\in\mathcal D$. The inclusion $(\mathsf K_{\mathrm{fr}})_{Q,P}\subseteq \mathsf{CP}_{Q,P}$ is automatic,
since compressions of completely positive maps are completely positive.
For the reverse inclusion, let $\psi\in \mathsf{CP}_{Q,P}$ be arbitrary and define its cut--pad extension
\[
\widehat\psi(T):=Q\,\psi(PTP)\,Q,\qquad T\in\mathcal B(\mathcal H).
\]
Then $\widehat\psi$ is normal and completely positive, and
$\mathrm{Ran}(\widehat\psi)\subseteq Q\mathcal B(\mathcal H)Q\cong \mathcal B(\ran Q)$ is finite-dimensional, so
$\widehat\psi\in\mathsf K_{\mathrm{fr}}$. Moreover $(\widehat\psi)_{Q,P}=\psi$. Hence
$\mathsf{CP}_{Q,P}\subseteq(\mathsf K_{\mathrm{fr}})_{Q,P}$, proving equality.

(c) Let $\Phi=\mathrm{id}_{\mathcal B(\mathcal H)}$. For any $(Q,P)\in\mathcal D$, the compression $\Phi_{Q,P}$ is
completely positive, hence $\Phi_{Q,P}\in \mathsf{CP}_{Q,P}$. By (b) this equals $(\mathsf K_{\mathrm{fr}})_{Q,P}$,
so $(\mathrm{id})_{Q,P}\in (\mathsf K_{\mathrm{fr}})_{Q,P}$ for all corners. On the other hand, $\mathrm{id}$ does not
have finite-dimensional range, hence $\mathrm{id}\notin\mathsf K_{\mathrm{fr}}$.
\end{proof}

\begin{remark}
  Proposition~\ref{prop:need-uw-closed} is consistent with the Corner-Closure Identity   (Theorem~\ref{thm:corner-closure}).
Although the cone of finite-rank maps $\mathsf K_{\mathrm{fr}}$ is not point-ultraweakly closed, its corner data is identical to that of the completely positive cone:
\[
(\mathsf K_{\mathrm{fr}})_{Q,P} = \mathsf{CP}_{Q,P}.
\]

Applying the reconstruction theorem to these corners yields
$\mathsf{CP}$. Hence, by the Corner-Closure Identity,
\[
\overline{\mathsf K_{\mathrm{fr}}}^{\,\mathrm{puw}}
=
\mathsf{CP}.
\]
Thus the finite-corner reconstruction recovers the point-ultraweak
closure of $\mathsf K_{\mathrm{fr}}$, rather than the nonclosed cone
$\mathsf K_{\mathrm{fr}}$ itself.

\end{remark}

\noindent \textit{Conflict of Interest Statement.}  There is no conflict of interest.
\medskip

\noindent \textit{Ethical Statement.}  Not applicable. 

\medskip
\noindent \textit{Informed Consent.} Not applicable.

\medskip
\noindent \textit{Funding Statement.} No funding was received for conducting this study.

\medskip
\noindent\textit{Data Availability Statement.} Data sharing not applicable to this article as no datasets were generated or analyzed during the current study.


\end{document}